\documentclass[11pt,a4paper]{article}

\usepackage[headinclude]{typearea}
\usepackage[english]{babel}           
\usepackage[T1]{fontenc}
\usepackage{scrtime}
\usepackage{amsmath,amsthm,amsfonts,amssymb}
\usepackage {color}                                  
\usepackage {pifont}                                 
\usepackage {multicol,multirow} 

\usepackage{authblk} 

\usepackage[numbers]{natbib}                                                     
\usepackage[normalem]{ulem}                                                     
\usepackage [scaled=.90]{helvet}                     
\usepackage {courier}                                
\usepackage {graphicx}                               
\usepackage {array}                                  
\usepackage {longtable}                              
\usepackage{fancyvrb}
\usepackage{enumerate} 
\usepackage{ifpdf}
\usepackage{caption}
\usepackage{mathrsfs}

\usepackage{setspace}

\usepackage{marvosym}

\usepackage{mathtools}
\mathtoolsset{showonlyrefs} 

\newtheorem{theorem}{Theorem}[section]

\theoremstyle{definition}

\newtheorem{algorithm}[theorem]{Algorithm}

\newcommand{\exclude}[1]{}

\definecolor{darkgreen}{rgb}{0,0.5,0}
\definecolor{lightgreen}{rgb}{0.5,0.9,0.5}
\definecolor{magenta}{rgb}{0.75,0,0.25}

\definecolor{violet}{rgb}{0.25,0,0.75}

\newcommand{\be}{\begin{equation}}
\newcommand{\ee}{\end{equation}}
\newcommand{\bea}{\begin{eqnarray}}
\newcommand{\eea}{\end{eqnarray}}
\newcommand{\beast}{\begin{eqnarray*}}
\newcommand{\eeast}{\end{eqnarray*}}
\newcommand{\bproof}{\begin{proof}}
\newcommand{\eproof}{\end{proof}}

\usepackage{algorithm}
\usepackage[noend]{algorithmic}
\newcommand{\Description}[1]{\relax}
\usepackage{xcolor,fancyhdr,manyfoot, url}
\usepackage{array,booktabs,amsmath,graphicx,fancyvrb,tabularx,longtable}
\usepackage{amsfonts, amsmath, amssymb, amsthm}
\graphicspath{{images/}}

\newcommand{\curvewidth}{0.98\textwidth} 
\newcommand{\gridwidth}{0.3\textwidth} 
\newcommand{\CMT}[1]{\COMMENT{#1}}

\makeatletter
\newcommand\footnoteref[1]{\protected@xdef\@thefnmark{\ref{#1}}\@footnotemark}
\makeatother

\begin{document}

\title{On the Closest Pair of Points Problem}



\author[1]{Martin Hitz
}

\author[2]{Michaela Hitz
}
\affil[1]{Department of Informatics Systems, University of Klagenfurt, Klagenfurt, Austria\\ ORCID:~0000-0002-8581-1720, martin.hitz@aau.at (corresponding author)}
\affil[2]{Department of Statistics, University of Klagenfurt, Klagenfurt, Austria\\ ORCID:~0000-0003-2987-2251, michaela.hitz
@aau.at}



\maketitle
\begin{abstract}
We introduce two novel algorithms for the problem of finding the closest pair in a cloud of $n$ points based on findings from mathematical optimal packing theory. Both algorithms are deterministic, show fast effective runtimes, and are very easy to implement. For our main algorithm, cppMM, we prove $O(n)$ time complexity for the case of uniformly distributed points. 
Our second algorithm, cppAPs, is almost as simple as the brute-force approach, but exhibits an extremely fast empirical running time, although its worst-case time complexity is also $O(n^2)$. 
We embed the new algorithms in a review of the most prominent contenders and empirically demonstrate their runtime behavior for problem sizes up to $n =$ 33,554,432 points observed in our C++ test environment. For large $n$, cppMM dominates the other algorithms under study. \\

\end{abstract}




\section{Introduction}

`Closest Pair of Points' (CPP) is a classical problem in computational geometry with relevant applications. For example, modern traffic-control systems, e.g.~for self-driving cars, need to know the two closest vehicles to detect potential collisions.

The problem can be formulated as follows:~Given a set $P =  \{p_1,\ldots, p_n \} = \{(p_1.x, p_1.y), \allowbreak  
\ldots, (p_n.x, p_n.y)\}$ of $n\in\mathbb{N}$ points in a compact set on the plane\footnote{Most algorithms can be generalized to points in a higher-dimensional space. For the sake of clarity of presentation, we consider the 2-dimensional case only.} (e.g., the unit square), find the pair of points that are closest together with respect to a metric $d_k(p, q) = \sqrt[k]{|p.x - q.x|^k + |p.y-q.y|^k}$ for any $k \geq 1$. In what follows, we always use the Euclidean distance $d_2(p, q) = \lVert(p, q)\rVert = \sqrt{(p.x - q.x)^2 + (p.y - q.y)^2)}$. 

The literature reports several solutions:~The brute force approach, which compares the distance of all pairs of points, has time complexity $\Theta(n^2)$, a divide-and-conquer algorithm with $O(n\log n)$ running time \cite{CormenEtAl2003}, or an $O(n\log n)$ plane-sweep algorithm \cite{HinrichsEtAl1988}. These are examples for optimal algorithms under the decision-tree model of computation which requires $\Omega(n\log n)$ steps for this problem \cite{PreparataShamos1985}. Rabin \cite{Rabin1976} proposed a probabilistic algorithm, later commented and slightly modified by Lipton \cite{Lipton2009}. Provided a hash table simulating a possibly huge two-dimensional matrix guarantees $O(1)$ access time to the matrix elements, the latter solves the problem in $O(n)$ time. For this, it exploits topological information by assigning all points to the cells of a grid the size of which depends on an estimate of the shortest distance that is computed in the brute-force manner from a random sample of all input points. Another probabilistic $O(n)$ algorithm put forward by Khuller and Matias shows some similarities to the Rabin-Lipton approach, in the sense that it also uses a grid to capture the topological distribution of the points, with a grid size which also depends on an approximation of the shortest distance computed from a random sample \cite{KhullerMatias1995}.

We revisit the CPP problem and present two novel algorithms that are easy to implement, deterministic, non-recursive, and yield very promising results.
Our main idea for our main contribution, algorithm cppMM, is based on findings from mathematical optimal packing theory. In particular, we make use of a proven upper bound $\bar{\delta}$ of the minimum distance $\delta$ of two points in a cloud of size $n$ for constructing a grid akin to those mentioned above \cite{SzaboEtAl2001}. By this we spare the step of approximating the minimal distance. Our approach leads to a significantly coarser and hence straight-forwardly implementable grid of size $O(n)$, while keeping the average case complexity linear for many cases.
The other algorithm, cppAPs, is a significantly accelerated and thus practically well usable version of the brute-force approach, augmented by a simple stopping condition akin to the `active window' used by Hinrichs et al. \cite{HinrichsEtAl1988}.

We present a comprehensive simulation study of our algorithms, comparing them to several algorithms from the literature.
On top, for cppMM, we prove linear time complexity for the case of uniformly distributed points.
All algorithms discussed were implemented in C++ and subjected to an empirical running time comparison using randomly generated point clouds from $2^{10} = 1024$ to $2^{25} = 33\,554\,432$ points.

\section{Classical algorithms}

\paragraph*{Comparing All Pairs}
The straightforward algorithm cppAP compares all $\binom{n}{2}$ pairs of points with each other and takes $\Theta(n^2)$ time, see Algorithm \ref{alg:cppAP}.

\begin{algorithm}[h]
\caption{\, cppAP$(P)\rightarrow <(p_{min}, q_{min}), \delta>$}
\begin{algorithmic}[1] 
\REQUIRE $P = \{p_1, \ldots, p_n\}, n \geq 2$ 
\ENSURE $\delta = \min\{\lVert(p, q) \rVert \mid p,q \in P, \, p\neq q\}$
\STATE $\delta \leftarrow \infty$
\FOR{$i \leftarrow 1$ \textbf{to} $n-1$} 
    \FOR{$j  \leftarrow i+1$ \textbf{to} $n$} 
        \IF{$\lVert (p_i, p_j) \rVert < \delta$}
            \STATE $\delta \leftarrow \lVert (p_i, p_j) \rVert$, cpp $ \leftarrow (p_i, p_j)$
        \ENDIF
    \ENDFOR
\ENDFOR
\RETURN $<\operatorname{cpp}, \delta>$ \CMT{return a structure consisting of two closest points and their distance $\delta$}
\end{algorithmic}
\label{alg:cppAP}
\end{algorithm}

Note that one could linearly speed up cppAP by replacing the Euclidean distance $\lVert (p, q) \rVert$ with its square, thus avoiding the calculation of the square root in lines 4 and 5, and returning the square root of $\delta$ in line 6.\footnote{In our experimental setting, this speed-up was about 7.5\%. This trick can be applied to several of the algorithms investigated. However, for our empirical studies, we always use the Euclidean distance proper, since some algorithms need it.}

\paragraph*{Divide-and-Conquer}
A faster approach is the Divide-and-Conquer algorithm; it takes $O(n \log n)$ time \cite{CormenEtAl2003}. The idea is to split the point set into two halves by the points' $x$-coordinates, recursively find the closest pair within each half, and take the minimum $\delta$ of the two resulting distances. Then, check if there is a pair of points between the two halves that is closer than the minimum distance in either of the halves, in which case, use that distance as the final result.

\begin{algorithm}[h]
\caption{\textbf{a} cppDC\textsubscript{0}$(P)\rightarrow\langle(p_{min},q_{min}),\delta\rangle$}\label{alg:cppDC}
\begin{algorithmic}[1] 
\REQUIRE $P = \{p_1, \ldots, p_n\}, n \geq 2$ 
\ENSURE $\delta = \min\{\lVert (p, q) \rVert \mid p,q \in P, \, p\neq q \}$
\STATE $X \leftarrow$ sort\textsubscript{x}$(P)$ 
\STATE $Y \leftarrow$ sort\textsubscript{y}$(P)$ 
\RETURN cppDC$(X, 1, n, Y)$ \CMT{see Algorithm \ref{alg:cppDC}b}
\end{algorithmic}
\end{algorithm}

In order to partition the plane in the divide phase, in a preparation step (cppDC\textsubscript{0}), the input set $P$ is sorted by the $x$-coordinates into an array $X = [(x_i, y_i) \in P  \mid x_i \leq x_j,  1 \leq i < j \leq n]$. In each recursion step, the divide phase splits the pertinent part of $X$ into two sub-arrays to the left and the right of its mid point, respectively. In the conquer phase, a $2\delta$ wide (vertical) stripe centered around that mid-point is searched for a possible pair of closer points. For each point in the left half of this stripe, only the points in the right half within a $y$-distance of at most $\pm\delta$ are checked. To support this process, a second array $Y = [(x_i, y_i) \in P \mid y_i \leq y_j,  1 \leq i < j \leq n]$ is computed in the preparation step, containing the points of $P$ in the order of increasing $y$-coordinates. Details are described in Algorithm \ref{alg:cppDC}.

\setcounter{algorithm}{1}
\begin{algorithm}[h]
\caption{\textbf{b} cppDC$(X, \operatorname{first}, \operatorname{last}, Y)\rightarrow <(p_{min}, q_{min}), \delta>$}
\begin{algorithmic}[1] 
\REQUIRE $n := \operatorname{last} - \operatorname{first} + 1, n = |Y|, n \geq 2$ 
\ENSURE $\delta = \min\{\lVert (p, q) \rVert \mid p,q \in P, \, p\neq q \}$
\IF{$n \leq 3$}
    \STATE $< \operatorname{cpp}, \delta> \leftarrow$ cppAP$(X[\operatorname{first}:\operatorname{last}])$ \CMT{brute force, see Algorithm \ref{alg:cppAP}}
\ELSE
    \STATE $m \leftarrow \lfloor(\operatorname{first} + \operatorname{last}) / 2\rfloor$
    \STATE $\operatorname{mid} \leftarrow (x_m, y_m) \in X$
    \STATE $Y_l \leftarrow [(x_i, y_i) \in Y \mid x_i \leq \operatorname{mid}.x]$ \CMT{order preserved}
    \STATE $Y_r \leftarrow [(x_i, y_i) \in Y \mid x_i > \operatorname{mid}.x]$ \CMT{order preserved}
    \STATE $<\operatorname{cpp}_l, d_l> \leftarrow$ cppDC$(X, \operatorname{first}, m, Y_l)$
    \STATE $<\operatorname{cpp}_r, d_r> \leftarrow$ cppDC$(X, m+1, \operatorname{last}, Y_r)$
    \IF{$d_l < d_r$}
        \STATE $\delta \leftarrow d_l$, $\operatorname{cpp} \leftarrow \operatorname{cpp}_l$
    \ELSE
        \STATE $\delta \leftarrow d_r$, $\operatorname{cpp} \leftarrow \operatorname{cpp}_r$
    \ENDIF
    \STATE $X' \leftarrow [(x_i, y_i) \in X \mid \operatorname{mid}.x - \delta \leq x_i \leq \operatorname{mid}.x]$  
    \STATE $Y' \leftarrow [(x_i, y_i) \in Y_r \mid x_i \leq \operatorname{mid}.x + \delta]$ \CMT{order preserved}
    \FOR{$p_1 \in X'$}
        \FOR{$p_2 \in [p \in Y' \mid |p_1.y - p.y| \leq 2\delta]$\footnotemark} 
            \IF{$\lVert (p_1, p_2) \rVert < \delta$}
                \STATE $\delta \leftarrow \lVert (p_1, p_2) \rVert$, $\operatorname{cpp} \leftarrow (p_1, p_2)$
            \ENDIF
        \ENDFOR
    \ENDFOR
\ENDIF
\RETURN $< \operatorname{cpp}, \delta>$
\end{algorithmic}
\end{algorithm}
\footnotetext{To compute this subset of $Y'$ we binary search the highest point on or below $p_1.y + \delta$ and proceed downwards until we reach the level of $\delta$ below $p_1.y$.}

\paragraph*{Plane-Sweep}
The plane-sweep approach by Hinrichs et al. \cite{HinrichsEtAl1988} is based on the general plane-sweep technique put forward by Shamos and Hoey \cite{ShamosHoey1976}.
This iterative algorithm moves a vertical line over the plane from left (no points at all considered yet) to right, processing a single point at a time until all points have been checked and the final solution for the whole set is obtained. This `sweeping' over the plane is orchestrated by a queue of points $xQueue$ ordered by their x-coordinates. 
The currently relevant points to compare the processed point with are maintained in a map $yTable$ based on y-coordinates, which associates the points' y-coordinates with their respective points. 

$xQueue$ is initialized with the set of points sorted according to their x-coordinates. The closest pair of points $cpp$ is initialized with the first two points in the queue, $\delta$ holding the respective distance, i.e.~the shortest distance found so far. The two points are also inserted into $yTable$.
In a loop over all points $current$ (starting with $q_3$, the third point of $xQueue$), we first remove all points $q_t$ from $yTable$ that lie left of $current$ and are further away than $\delta$ – thus outside of the so-called `active window'. The remaining points in $yTable$ with y-coordinates between $current.y - \delta$ and $current.y + \delta$ are compared to $current$, yielding a new closest pair as well as an update of $\delta$, if the distance is less than $\delta$. After the complete rectangular $\delta \times 2\delta$ environment to the left of $current$ has been checked in this way, point $current$ is entered into $yTable$, and the loop proceeds to the next point.

\begin{algorithm}[h]
\caption{cppPS$(P)\rightarrow <(p_{min}, q_{min}), \delta>$}\label{alg:cppPS}
\begin{algorithmic}[1] 
\REQUIRE $P = \{p_1, \ldots, p_n\}, n \geq 2$ 
\ENSURE $\delta = \min\{\lVert (p, q) \rVert \mid p,q \in P, \, p\neq q \}$
\STATE xQueue $=[q_1,\dots,q_n] \leftarrow$ sort\textsubscript{x}$(P)$ \CMT{xQueue holds the points $q_i$ sorted by their x-coordinates}
\STATE yTable $\leftarrow []$ \CMT{yTable maps the y-coordinates of points to the points themselves}
\STATE cpp $ \leftarrow (q_1, q_2)$
\STATE $\delta \leftarrow \lVert (q_1, q_2) \rVert$
\STATE yTable.insert($q_1$), yTable.insert($q_2$)
\STATE $t \leftarrow 0$ \CMT{$\widehat{=}$ pointer `tail' into xQueue \cite{HinrichsEtAl1988}}
\FOR{$i \leftarrow 3$ \textbf{to} $n$} 
    \STATE current $\leftarrow$ $q_i$
    \WHILE{$q_t.x \leq$ current.$x - \delta$}
        \STATE yTable.delete($q_t$)
        \STATE $t \leftarrow t+1$ 
    \ENDWHILE
    \FOR{$p \in [p \in \operatorname{yTable} \mid |p.y - \operatorname{current}.y| < \delta]$}
        \IF {$p \neq \operatorname{current} \land \lVert (\operatorname{current}, p) \rVert < \delta$}
            \STATE $\delta \leftarrow \lVert (\operatorname{current}, p) \rVert$, $\operatorname{cpp} \leftarrow (\operatorname{current}, p)$
        \ENDIF
    \ENDFOR
    \STATE yTable.insert(current)
\ENDFOR
\RETURN $< \operatorname{cpp}, \delta>$
\end{algorithmic}
\end{algorithm}

As explained by Hinrichs et al. \cite{HinrichsEtAl1988}, the running time required to initialize $xQueue$ is $O(n \log n)$ and $O(\log n)$ for entering (or deleting) a y-coordinate value into $yTable$. Since the latter is done once for each point in an iteration over all $n$ points, the overall running time results in $O(n \log n)$.

\paragraph*{Rabin-Lipton}
Rabin \cite{Rabin1976} showed how to solve CPP in linear time using a probabilistic algorithm that has been commented on and slightly adapted by Lipton \cite{Lipton2009}. The `Rabin-Lipton approach' cppRL proceeds as follows:~Take a random sample of $\sqrt{n}$ points from $P$ and compute the minimum distance $d$ of points in this sample by cppAP. This takes $O(\sqrt{n}^2) = O(n)$ steps. Then divide the unit square into a grid of square cells sized $d \times d$. Although for very small values of $d$, the number $\lceil 1/d \rceil \times \lceil 1/d \rceil$ of cells can be very large ($\gg n$), only at most $n$ cells will contain one or more points of $P$. It is necessary to organize these up to $n$ `filled' cells in a way that it is possible to access them in sequence in $O(n)$ steps, and individually in $O(1)$ steps. A hash table of size $O(n)$ with keys composed of both grid indices $i, j$ will do the job. Using this data structure, for each point $p\in P$, check all points which potentially have a distance at most $d$ to $p$. These must reside either in the same cell as $p$ or in one of the at most eight adjacent filled cells. Thus, for each of the $n$ points in $P$ it suffices to compute the distances to the points of at most nine cells.\footnote{For symmetry reasons, it suffices to check the cell itself and half of its neighbours, see Section \ref{sec:cppMM} for details.\label{FN:half}} \footnote{Note that the original approach calls for two separate phases:~First, compute the closest pair in every filled grid cell. Second, compute the closest pair for every pair of adjacent filled cells, see \citep{Lipton2009}. Although our implementation does basically the same thing, we believe that our scheme `\textit{compute the closest pair within the cell and its filled neighbours}' is easier to understand.\label{FN:original}}
Details are described in Algorithm~\ref{alg:cppRL}. As Rabin shows, this can be done in constant time, yielding an overall average case time complexity of $O(n)$. Note that Lipton argues that for the proof of the result, \textit{distances} should be sampled instead of \textit{points} \cite{Lipton2009}. In our experiments, distance sampling turned out to be also empirically faster than point sampling -- so we present the running times for cppRL with distance sampling below.

\begin{algorithm}
\caption{cppRL$(P)\rightarrow <(p_{min}, q_{min}), \delta>$}\label{alg:cppRL} 
\begin{algorithmic}[1] 
\REQUIRE $P = \{p_1, \ldots, p_n \} \subset [0, 1]^2, n \geq 2$  
\ENSURE $\delta = \min\{\lVert (p, q) \rVert \mid p,q \in P, \, p\neq q\}$
\STATE $S \leftarrow$ randomSample$(P, \lfloor\sqrt{n}\rfloor)$ \CMT{take random sample of size $\lfloor\sqrt{n}\rfloor$}
\STATE $< \operatorname{cpp}, d> \leftarrow$ cppAP$(S)$ \CMT{brute force, see Algorithm \ref{alg:cppAP}}
\STATE $g \leftarrow \lceil 1/d \rceil$
\STATE $G := g \times g$ grid of empty buckets \CMT{$G$ uses a hash table of size $O(n)$}
\FOR{$p \in P$}
    \STATE $i \leftarrow \lfloor p.x / d \rfloor + 1$, $j \leftarrow \lfloor p.y / d \rfloor + 1$
    \STATE $G_{i, j} \leftarrow G_{i, j} \cup \{p\}$ \CMT{$G_{i, j}$ is (now) a `filled' bucket}
    \STATE visited$(G_{i, j}) \leftarrow$ \textbf{false} \CMT{only necessary if visited($G_{i, j}$) is not initialized in line 4}
\ENDFOR
\STATE $\delta \leftarrow d$
\FOR{$p \in P$}
    \STATE $i \leftarrow \lfloor p.x / d \rfloor + 1$, $j \leftarrow \lfloor p.y / d \rfloor + 1$ \CMT{found `filled' bucket $G_{i, j}$}
    \IF {$\neg $visited$(G_{i, j})$}
        \STATE visited$(G_{i, j}) \leftarrow$ \textbf{true} \CMT{make sure to visit the bucket only once}
        \FOR{$p_1 \in G_{i, j}$}
            \FOR{$p_2 \in (G_{i, j} \setminus\{p_1\}) \cup G_{i, j+1} \cup G_{i+1, j-1} \cup G_{i+1, j} \cup G_{i+1, j+1}$\footnoteref{FN:half} \footnoteref{FN:original} ~\footnoteref{FN:nonexHash}}
                \IF{$\lVert (p_1, p_2) \rVert < \delta$}
                    \STATE $\delta \leftarrow \lVert (p_1, p_2) \rVert$, $\operatorname{cpp} \leftarrow (p_1, p_2)$
                \ENDIF
            \ENDFOR
        \ENDFOR
    \ENDIF
\ENDFOR
\RETURN $< \operatorname{cpp}, \delta>$
\end{algorithmic}
\end{algorithm}

\paragraph*{Khuller-Matias}
Khuller and Matias present another randomized algorithm with time complexity $O(n)$ that is based on a filtering approach (sieve) \cite{KhullerMatias1995}. The sieving process starts with a copy $P_1$ of the entire point set and iteratively proceeds as follows:~In iteration $i$, compute an approximation of $\delta$ by picking a random point $x_i$ from $P_i$ and finding the minimum distance $d(x_i)$ from that point to the other points in $P_i$. Then, construct a grid of size $d(x_i)/3$ and assign the points of $P_i$ to its cells. Using this grid, define $P_{i+1}$ by taking all points from $P_i$ except those that happen to reside alone in their grid cell, surrounded by eight empty neighbouring cells. This iteration terminates when $P_{i+1}$ becomes empty. In the final phase of the algorithm, create a grid of the size of the last computed $d(x_i)$ and assign the points of $P$ to its cells. Then, for each of these points, find the closest partner-point within its grid cell and its eight neighbouring cells. From these $n$ `local' closest pairs, select the one with smallest distance as the overall result.
Details are described in Algorithm \ref{alg:cppKM}. Similar to cppRL, the grids used are implemented as hash tables to guarantee $O(n)$ space complexity.

\begin{algorithm}
\caption{cppKM$(P)\rightarrow <(p_{min}, q_{min}), \delta>$} \label{alg:cppKM} 
\begin{algorithmic}[1] 
\REQUIRE $P = \{p_1, \ldots, p_n \} \subset [0, 1]^2, n \geq 2$ 
\ENSURE $\delta = \min\{\lVert (p, q) \rVert \mid p,q \in P, \, p\neq q\}$
\STATE $i \leftarrow 1$, $P_1 \leftarrow P$ 
\WHILE{$P_i \neq \varnothing$}
    \STATE $x_i \leftarrow$ randomSample$(P_i, 1)$ \CMT{select a single random point from $P_i$}
    \STATE $d \leftarrow 2$
    \FOR{$p \in P_i\setminus\{x_i\}$}
        \IF{$\lVert (x_i, p) \rVert < d$}
            \STATE $d \leftarrow \lVert (x_i, p) \rVert$
        \ENDIF
    \ENDFOR
    \STATE $g \leftarrow d/3$
    \STATE $G := g \times g$ grid of empty buckets \CMT{$G$ uses a hash table of size $O(n)$}
    \FOR{$p \in P_i$}
        \STATE $i \leftarrow \lfloor p.x / g \rfloor + 1$, $j \leftarrow \lfloor p.y / g \rfloor + 1$
        \STATE $G_{i, j} \leftarrow G_{i, j} \cup \{p\}$
    \ENDFOR
    \STATE $P_{i+1} \leftarrow P_i$ \CMT{prepare next iteration}
    \FOR{$p \in P_i$}
        \STATE $i \leftarrow \lfloor p.x / g \rfloor + 1$, $j \leftarrow \lfloor p.y / g \rfloor + 1$
        \STATE $N \leftarrow G_{i-1, j-1}\cup G_{i-1, j}\cup G_{i-1, j+1}\cup G_{i, j-1} \cup G_{i, j}\cup G_{i, j+1}\cup G_{i+1, j-1}\cup G_{i+1, j}\cup G_{i+1, j+1}$ \footnotemark
        \IF{$|N| = 1$} 
            \STATE $P_{i+1} \leftarrow P_{i+1} \setminus \{p\}$\CMT{remove $p$ from next iteration}
        \ENDIF
    \ENDFOR
    \STATE $i \leftarrow i+1$
\ENDWHILE
\STATE $G := d \times d$ grid of empty buckets
\FOR{$p \in P$}
    \STATE $i \leftarrow \lfloor p.x / d \rfloor + 1$, $j \leftarrow \lfloor p.y / d \rfloor + 1$
    \STATE $G_{i, j} \leftarrow G_{i, j} \cup \{p\}$
\ENDFOR
\STATE $\delta \leftarrow d$
\FOR{$p_1 \in P$}
    \STATE $i \leftarrow \lfloor p_1.x / d \rfloor + 1$, $j \leftarrow \lfloor p_1.y / d \rfloor + 1$
    \STATE $N \leftarrow G_{i-1, j-1}\cup G_{i-1, j}\cup G_{i-1, j+1}\cup G_{i, j-1} \cup  (G_{i, j}\setminus\{p_1\})\cup G_{i, j+1}\cup G_{i+1, j-1}\cup G_{i+1, j}\cup G_{i+1, j+1}$\footnoteref{FN:nonexHash}
    \FOR{$p_2\in N$}
        \IF{$\lVert (p_1, p_2) \rVert < \delta$}
            \STATE $\delta \leftarrow \lVert (p_1, p_2) \rVert$, $\operatorname{cpp} \leftarrow (p_1, p_2)$
        \ENDIF
    \ENDFOR
\ENDFOR
\RETURN $< \operatorname{cpp}, \delta>$
\end{algorithmic}
\end{algorithm}
\footnotetext{The hash table used to implement $G$ assures that empty buckets are returned for non-existing grid cell indices, so we do not have to deal with boundary conditions here.\label{FN:nonexHash}}

\section{Novel algorithms}

In this section we present two novel algorithms -- an improvement of cppAP that we call cppAPs (short for `All Pairs sorted') and our main result, cppMM.
Note that both algorithms are completely deterministic.

\subsection{Improving the comparing-all-pairs approach:~\textbf{cppAPs}} \label{sec:cppAPs}

It may seem strange at first glance to work on improving an $O(n^2)$ approach when asymptotically faster algorithms exist. However, our all-pairs-sorted algorithm cppAPs is extremely easy to implement, and, in our experimental setting, turns out to be the fastest of all tested algorithms up to $n=2^{21}$ in case of uniformly distributed test-points (see Fig.~\ref{fig:UD-RT-C}) and up to $n=2^{23}$ in case of truncated normally distributed test-points (see Fig.~\ref{fig:ND-RT-C}). Therefore, cppAPs seems to be a good alternative for practical purposes, even though its worst-case time complexity is still $O(n^2)$.

The idea behind cppAPs is as follows:~First, sort $P$ by the $x$-coordinates of the points (yielding a point-list $X$). Then proceed similar as in cppAP, but in both \textbf{for}-loops test all possible pairs of points in the sequence of ascending $x$-coordinates. If for a given point $q_i \in X$, the nested \textbf{for}-loop (which starts at point $q_{i+1}$) reaches a point $q_j$ with an $x$-coordinate farther apart from $q_i$ than the current minimum distance $d$ (e.g.~$q_j.x - q_i.x \leq d$), it makes no sense to check this point and the subsequent points, because the current minimum will not change anymore. Note that this is related to the approach to concentrate on the `active window' by Hinrichs et al. \cite{HinrichsEtAl1988}. Therefore, the nested \textbf{for}-loop may terminate after very few -- but, unfortunately, possibly $O(n)$ -- iterations. Details are described in Algorithm \ref{alg:cppAPs}.

\begin{algorithm}[h]
\caption{cppAPs$(P)\rightarrow <(p_{min}, q_{min}), \delta>$}\label{alg:cppAPs}
\begin{algorithmic}[1] 
\REQUIRE $P = \{p_1, \ldots, p_n\}, n \geq 2$  
\ENSURE $\delta = \min\{\lVert (p, q) \rVert \mid p,q \in P, \, p\neq q\}$
\STATE $Q=[q_1,\dots,q_n] \leftarrow$ sort\textsubscript{x}$(P)$ 
\STATE $\delta \leftarrow \infty$
\FOR{$i \leftarrow 1$ \textbf{to} $n-1$} 
    \FOR{$j  \leftarrow i+1$ \textbf{to} $n$ \textbf{while} $q_j.x - q_i.x \leq \delta$}
        \IF{$\lVert (q_i, q_j) \rVert < \delta$}
            \STATE $\delta \leftarrow \lVert (q_i, q_j) \rVert$, $\operatorname{cpp} \leftarrow (q_i, q_j)$
        \ENDIF
    \ENDFOR
\ENDFOR
\RETURN $<\operatorname{cpp}, \delta>$ 
\end{algorithmic}
\end{algorithm}

In order to illustrate the behavior of the nested \textbf{for}-loop (lines 4--6 in Algorithm \ref{alg:cppAPs} and lines 3--5 in Algorithm \ref{alg:cppAP}), let us compare the increase of the number of its iterations $I_2(n)$ when doubling the input size $n$ in algorithms cppAP and cppAPs. In Figure \ref{fig:for2}, the factor $I_2(n)/I_2(n/2)$ is plotted from $n=2^{11}$ to $n=2^{25}$ for cppAP and cppAPs, respectively. We observe that cppAP always doubles the number of iterations of the nested \textbf{for}-loop when doubling $n$, while cppAPs approaches that factor 2 only slowly and only for very large values of $n$. Thus cppAPs, while exhibiting asymptotic time complexity of $O(n^2)$, shows close to linear runtime until very large values of~$n$.
For many practical purposes, cppAPs is therefore a very good choice.

\begin{figure}[h]
    \centering
    \includegraphics[width=\curvewidth]{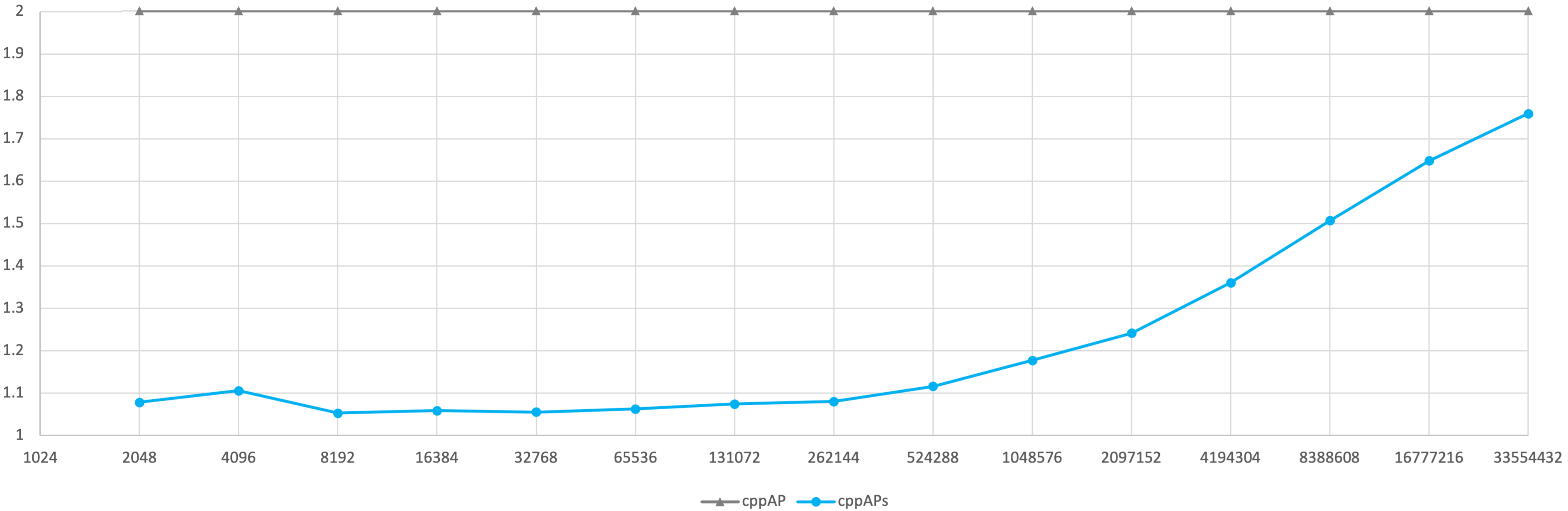}
    \caption{Number of iterations $I_2(n)$ of the nested \textbf{for}-loop of cppAP and cppAPs for $n$ points divided by $I_2(n/2)$. Log scale on the abscissa.}
    \Description{Growth of the number of iterations of the nested for-loop for cppAPs (rising from factor 1.1. to 1.75) vs. cppAP (always factor 2).}
\label{fig:for2}
\end{figure}

\subsection{A deterministic linear complexity algorithm:~\textbf{cppMM}} \label{sec:cppMM}

For constructing our main algorithm, we adopt the idea by Rabin of dividing the unit square into a grid of cells and searching these cells and their neighbours for the closest pair. In cppRL a simple search of their grid is not possible in linear time due to the very small grid size. Our idea is to make the grid as large as possible but still small enough such that not too many comparisons will be needed for each cell. Hence, we require an upper bound as sharp as possible for the maximum of the minimal distance. We will make use of the following result from mathematical optimal packing theory, which originates from Szabó et al. \cite{SzaboEtAl2001}.

\begin{theorem}[\text{\cite[Theorem 2]{SzaboEtAl2001}}]
The maximum of the minimal distance is not greater than
\begin{equation}\label{eq:upperbound}
 \bar{\delta} = \min(U_1(n), U_2(n)),
\end{equation} 
with
\begin{equation*}
U_1(n) = \frac{2}{\sqrt{n\pi+C_n(\sqrt{3} - \frac{\pi}{2}) + (4\lfloor \sqrt{n}\rfloor - 2)(2 - \frac{\pi}{2})}-2},
\end{equation*}
where
\begin{equation*}
  C_n=\begin{cases}
    n - 2, & \text{if $3 \leq n \leq 6$,}\\
    n - 1, & \text{if $7 \leq n \leq 9$,}\\
    3\lfloor \frac{n}{2} \rfloor - 5 + n\bmod{2}, & \text{otherwise}
  \end{cases}
\end{equation*}
and
\begin{equation*}
U_2(n) = \frac{1 + \sqrt{1+(n-1)\frac{2}{\sqrt{3}}}}{n - 1}.
\end{equation*}
\end{theorem}
Since $C_n\in O(n)$, we have that $U_1(n)\in O(\frac{1}{\sqrt{n}})$ as well as $U_2(n)\in O(\frac{1}{\sqrt{n}})$, showing that $\bar\delta\in O(\frac{1}{\sqrt{n}})$.

We use this result to divide the unit square into a $\lceil 1/\bar{\delta} \rceil \times \lceil 1/\bar{\delta} \rceil$ grid $G$ with cell size $\bar{\delta} \times \bar{\delta}$ and assign each point $p \in P$ to the respective cell. This is similar to cppRL, but employs a grid with only $O(n)$ cells which can be directly implemented as a two-dimensional array without jeopardizing linear running time nor linear space complexity. Other advantages of our algorithm are that we spare the precomputation step from cppRL and that this makes the algorithm deterministic, while cppRL is a randomized algorithm.

To find the closest pair, we iterate over all points in the grid, checking all potential partner-points within the same grid cell $G_{i,j}$, within the right neighbouring cell $G_{i,j+1}$, within the cell $G_{i+1,j-1}$ diagonally left underneath, within the cell $G_{i+1,j}$ directly underneath, and within the cell $G_{i+1,j+1}$ diagonally right underneath (whenever they exist). 
Figure \ref{fig:Grid} illustrates this search for the cell $G_{2,4}$.

\begin{figure}[h]
    \centering
    \includegraphics[width=\gridwidth]{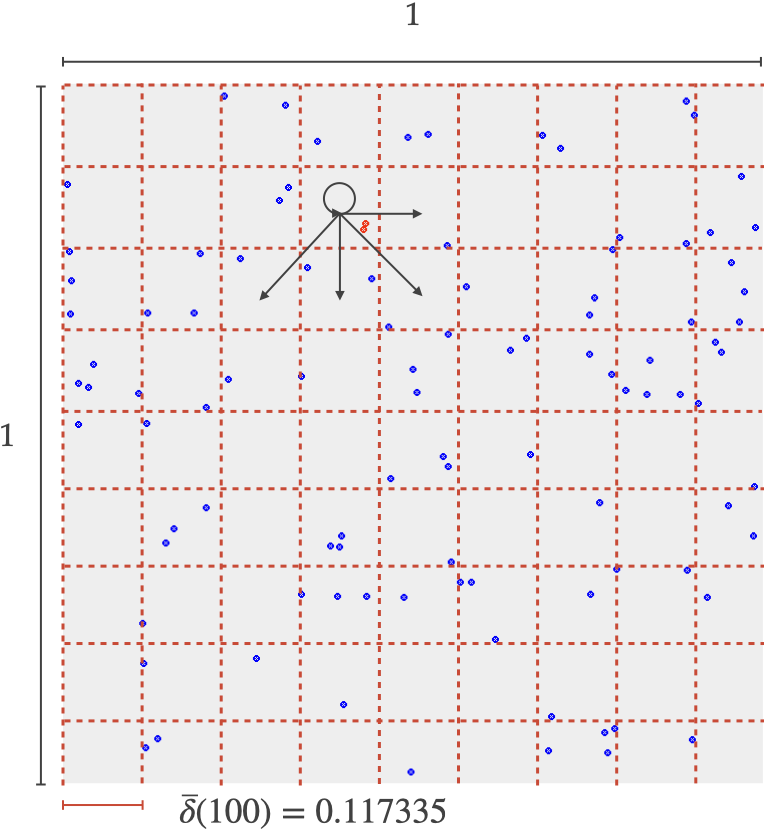}
    \caption{For $n = 100$ points in the unit square, the upper bound for the minimum distance is $\bar{\delta}=0.117$. Since $1/\bar{\delta}=8.523$, we construct a $ 
    9\times9$ grid with smaller residual grid cells on the right and bottom border of the square. A closest pair of points must either lie entirely in one of the grid cells, or one point resides in a cell and the other point in one of the eight neighbouring cells. Due to symmetry, it suffices to check only four of its neighbours of each cell, as illustrated for the case of cell $G_{2,4}$.}
    \Description{Arrows illustrate that the points in cell 2/4 are compared with each other and with the points in cells 2/5, 3/3, 3/4, and 3/5.}
    \label{fig:Grid}
\end{figure}

All partner-points in other cells have either already been checked during previous iterations or are too far apart to contain viable partner-points for a minimal distance due to the definition of the grid; that is, the final closest pair will have a distance not more than the cell size $\bar{\delta}$.
Details of the `Maximum-Minimal-distance‘ algorithm cppMM are described in Algorithm \ref{alg:cppMM}.

\begin{algorithm}[h] 
\caption{cppMM$(P)\rightarrow <(p_{min}, q_{min}), \delta>$}
\begin{algorithmic}[1] 
\REQUIRE $P = \{p_1, \ldots, p_n \} \subset [0, 1]^2, n \geq 2$  
\ENSURE $\delta = \min\{\lVert (p, q) \rVert \mid p,q \in P, \, p\neq q \}$
\STATE $\bar{\delta} \leftarrow$ upper bound from Equation \eqref{eq:upperbound}
\STATE $w \leftarrow \max\{p.x \mid p\in P\} - \min\{p.x \mid p\in P\}$
\STATE $h \leftarrow \max\{p.y \mid p\in P\} - \min\{p.y \mid p\in P\}$  
\STATE $\bar{\delta} \leftarrow \bar{\delta} \times \sqrt{w \times h}$ \CMT{scale $\bar{\delta}$ to size of actual point cloud}
\STATE $g \leftarrow \lceil 1/\bar{\delta} \rceil$
\STATE $G := g \times g$ grid of empty buckets
\FOR{$p \in P$}
    \STATE $i \leftarrow \lfloor p.x / \bar{\delta} \rfloor + 1$, $j \leftarrow \lfloor p.y / \bar{\delta} \rfloor + 1$
    \STATE $G_{i, j} \leftarrow G_{i, j} \cup \{p\}$ 
\ENDFOR
\STATE $\delta \leftarrow \bar{\delta}$
\FOR{$i \leftarrow 1$ \textbf{to} $g$}
    \FOR{$j \leftarrow 1$ \textbf{to} $g$}
        \FOR{$p_1 \in G_{i, j}$}
            \FOR{$p_2 \in  G_{i, j+1} \cup (G_{i, j} \setminus\{p_1\}) \cup G_{i+1, j-1} \cup G_{i+1, j} \cup G_{i+1, j+1}$\footnotemark}
                \IF{$\lVert (p_1, p_2) \rVert < \delta$}
                   \STATE $\delta \leftarrow \lVert (p_1, p_2) \rVert$, $\operatorname{cpp} \leftarrow (p_1, p_2)$
                \ENDIF
            \ENDFOR
        \ENDFOR
    \ENDFOR
\ENDFOR
\RETURN $< \operatorname{cpp}, \delta>$
\end{algorithmic}
\label{alg:cppMM}
\end{algorithm} 
\footnotetext{For the sake of brevity of presentation we do not mention boundary cases here. Non-existing grid cells are defined as empty sets.}

The correctness of the algorithm follows from the fact that it is a straightforward speed-up version of cppAP, excluding most unnecessary comparisons in the nested \textbf{for}-loop of Algorithm \ref{alg:cppAP} (lines 3--5) by using the `topologically guided' loop on lines 14--16 of Algorithm \ref{alg:cppMM}. 

\section{Performance analysis}
In this section we present a comprehensive empirical analysis, where we compare the performance of cppMM to the other algorithms.  We implemented all algorithms in C++.\footnote{For the review, the code can be downloaded from //www.aau.at/en/isys/ias/research/ and will be published under an appropriate open-source license when the paper is published.} The run-time comparisons were computed on an Apple M3 based MacBook Air with eight cores and 16 GB memory under macOS Sequoia 15.3, using the \verb|<time.h>| library function \verb|clock_gettime()|. The hash table for the implementation of the grids needed for cppRL and cppKM provides for $2^{ \lceil\log_2{3n}\rceil}$ entries and employs external collision chains with a typical mean length of $0.12$ (for non-empty entries).

The point sets for our test series were generated using out of the box pseudo-random number generators from the C++ library (\verb|std::default_random_engine| and\newline
\verb|std::normal_distribution<double>|). We analyze uniformly distributed and truncated\footnote{As the two-dimensional normal distribution lives on the whole of $\mathbb{R}^2$.} normally distributed point sets in the unit square.
For the convenience of the reader, we recall the following.
\paragraph*{Truncated normal distribution}
Let $f_{\mathcal{N}_2}$ be the density of the two-dimensional normal distribution with parameters $(\mu,\Sigma)$, that is for all $x\in\mathbb{R}^2$,
\begin{equation*}
f_{\mathcal{N}_2}(x)=\frac{1}{2\pi\sqrt{\det(\Sigma)}}\cdot\exp\Big(-\frac{1}{2} (x-\mu)^\top \Sigma^{-1} (x-\mu)\Big).
\end{equation*}
The density $f_{\mathcal{TN}_2}$ of its truncated version is chosen so that all sample points stay within the unit square, that is for all $x\in[0,1]^2$ we have
\begin{equation*}
f_{\mathcal{TN}_2}(x)= \frac{1}{\int_{[0,1]^2} f_{\mathcal{N}_2}(y) dy} \cdot f_{\mathcal{N}_2}(x).
\end{equation*}
Note that the first factor in $f_{\mathcal{TN}_2}$ above scales the function such that together with the choice $f_{\mathcal{TN}_2}(x)=0$ for all
$x\in\mathbb{R}^2\setminus[0,1]^2$, $f_{\mathcal{TN}_2}$ is a density.\footnote{
 Note that neither is $\mu$ in general the expectation, nor is $\Sigma$ the covariance matrix of $\mathcal{TN}_2(\mu,\Sigma)$.}\\
 
We sample from a two-dimensional truncated normal distribution $\mathcal{TN}_2(\mu,\Sigma)$ via the acceptance-rejection method.
In our implementation we choose the parameters $\mu=(0.5,0.5)^\top$ to center the point cloud within the unit square and $\Sigma=\text{diag}_2(0.04,0.04)$ to achieve a noteworthy difference to the uniform distribution (cf.~Figure \ref{fig:UDND1024}). For convenience, denote $\sigma = \sqrt{\Sigma_{11}} = \sqrt{\Sigma_{22}}$. Note that the choice $\sigma=0.2=\sqrt{0.04}$ leads to a variance in each coordinate of the truncated normal distribution of $0.036$, hence only a little smaller than $0.04$. In this special case, due to symmetry of the density, the expectation is the same as in the non-truncated case.

We start with $n = 2^{10} = 1024$ points and double the input size up to $n = 2^{25} = 33\,554\,432$ points\footnote{Due to its extensive running time, we stop cppAP earlier, since it is already obvious that it is overtaken by all other algorithms.}. Figure \ref{fig:UDND1024} illustrates typical instances of the problem with $n=1024$.

\begin{figure}[h]
    \centering
    \includegraphics[width=\gridwidth]{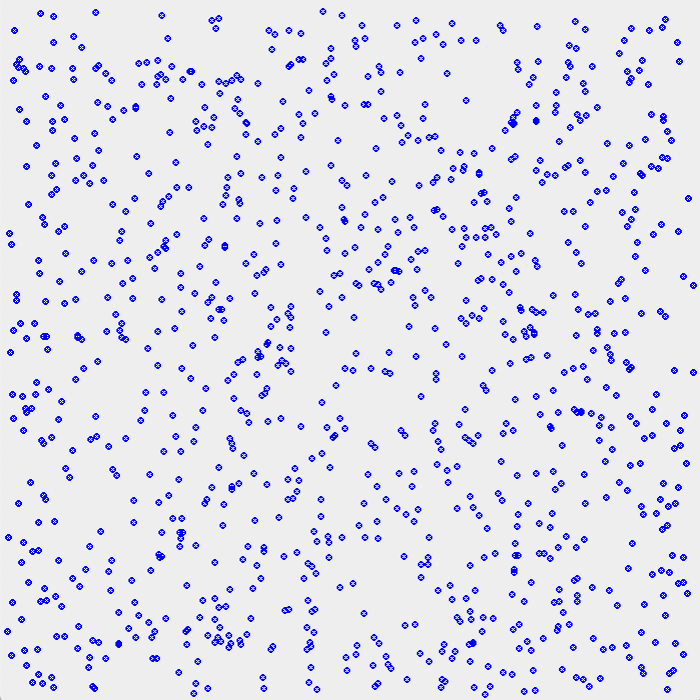}
    \includegraphics[width=\gridwidth]{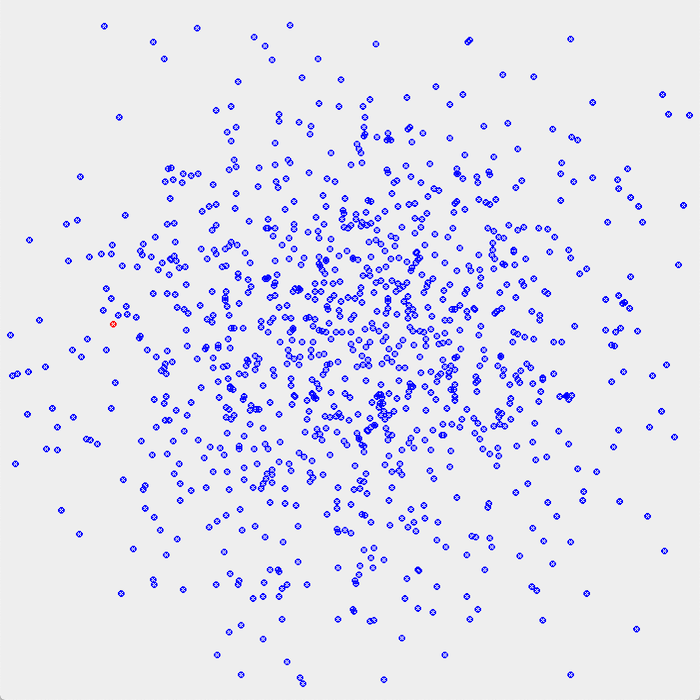}
    \caption{Typical test-sets with 1024 pseudo-random points:~uniform distribution (left) and truncated normal distribution with $\mu = (0.5, 0.5)^\top$ and $\sigma = 0.2$ (right).}
    \Description{The left point cloud contains equally scattered points, while the points of the right cloud are concentrated in the center.}
    \label{fig:UDND1024}
\end{figure}

Tables \ref{tab:UD-RT} and \ref{tab:ND-RT} as well as Figures \ref{fig:UD-RT-C} and \ref{fig:ND-RT-C} show the accumulated running time in seconds for ten different randomly drawn CPP problem instances for each value of $n$.

\begin{table}[h]
\scriptsize
    \centering
    \begin{tabular}{|r|r|r|r|r|r|r|r|r|}
        \hline
        $k$ & $n = 2^k$ & AP & APs & DC & KM & MM & PS & RL \\
        \hline
        10 & 1,024 & 4,711 & 103 & 388 & 529 & 132 & 531 & 172 \\
        11 & 2,048 & 10,625 & 112 & 497 & 643 & 168 & 642 & 222 \\
        12 & 4,096 & 42,020 & 222 & 1,055 & 1,268 & 339 & 1,283 & 526 \\
        13 & 8,192 & 172,488 & 428 & 2,271 & 2,545 & 666 & 2,538 & 1,014 \\
        14 & 16,384 & 706,161 & 879 & 4,840 & 5,037 & 1,320 & 5,058 & 1,850 \\
        15 & 32,768 & 2,826,759 & 1,894 & 10,754 & 11,425 & 2,776 & 10,857 & 3,923 \\
        16 & 65,536 & 11,339,138 & 3,704 & 22,701 & 23,195 & 5,542 & 21,104 & 8,032 \\
        17 & 131,072 & 45,350,724 & 7,908 & 45,989 & 47,616 & 11,499 & 43,420 & 18,113 \\
        18 & 262,144 &  & 17,865 & 91,294 & 127,298 & 25,853 & 94,709 & 53,226 \\
        19 & 524,288 &  & 43,761 & 182,940 & 364,034 & 59,576 & 213,286 & 158,652 \\
        20 & 1,048,576 &  & 115,600 & 369,691 & 937,741 & 161,351 & 499,074 & 400,559 \\
        21 & 2,097,152 &  & 334,576 & 748,543 & 2,186,129 & 421,355 & 1,200,979 & 811,589 \\
        22 & 4,194,304 &  & 1,092,343 & 1,532,403 & 4,772,536 & 1,016,446 & 3,013,602 & 1,791,181 \\
        23 & 8,388,608 &  & 3,891,827 & 3,231,553 & 9,825,886 & 2,159,184 & 8,182,573 & 4,043,268 \\
        24 & 16,777,216 &  & 14,485,259 & 6,829,380 & 19,934,263 & 4,552,791 & 23,970,723 & 8,831,863 \\
        25 & 33,554,432 &  & 55,736,141 & 14,728,577 & 40,933,580 & 9,533,593 & 75,796,244 & 16,009,355 \\
        \hline
    \end{tabular}
    \caption{Average running times of Algorithms 1-7 for 10 instances for each problem size $n$ in microseconds; uniformly distributed points.}
    \label{tab:UD-RT}
\end{table}

\begin{table}[h]
\scriptsize
    \centering
    \begin{tabular}{|r|r|r|r|r|r|r|r|r|}
        \hline
        $k$ & $n = 2k$ & AP & APs & DC & KM & MM & PS & RL \\
        \hline
        10 & 1,024 & 2,776 & 68 & 257 & 343 & 124 & 346 & 112 \\
        11 & 2,048 & 10,962 & 134 & 543 & 694 & 255 & 690 & 225 \\
        12 & 4,096 & 43,832 & 271 & 1,167 & 1,404 & 508 & 1,437 & 506 \\
        13 & 8,192 & 176,529 & 523 & 2,489 & 3,112 & 1,024 & 2,804 & 996 \\
        14 & 16,384 & 708,588 & 1,031 & 5,376 & 6,154 & 2,082 & 5,614 & 2,072 \\
        15 & 32,768 & 2,838,028 & 2,016 & 11,708 & 12,346 & 4,149 & 11,081 & 3,909 \\
        16 & 65,536 & 11,346,148 & 4,201 & 24,289 & 23,422 & 8,321 & 22,600 & 8,110 \\
        17 & 131,072 & 45,399,074 & 8,721 & 51,635 & 48,254 & 16,864 & 46,884 & 17,161 \\
        18 & 262,144 &  & 18,577 & 107,927 & 127,070 & 35,729 & 98,483 & 48,199 \\
        19 & 524,288 &  & 42,053 & 231,625 & 371,522 & 78,180 & 214,898 & 149,684 \\
        20 & 1,048,576 &  & 102,855 & 500,531 & 950,452 & 199,526 & 491,998 & 376,713 \\
        21 & 2,097,152 &  & 261,540 & 1,098,059 & 2,121,538 & 511,603 & 1,132,139 & 918,130 \\
        22 & 4,194,304 &  & 725,205 & 2,376,220 & 5,010,292 & 1,227,949 & 2,653,479 & 1,825,625 \\
        23 & 8,388,608 &  & 2,300,839 & 5,179,681 & 9,617,666 & 2,669,723 & 6,756,197 & 3,799,797 \\
        24 & 16,777,216 &  & 7,765,041 & 11,717,941 & 21,658,310 & 5,820,174 & 17,971,181 & 7,784,499 \\
        25 & 33,554,432 &  & 27,665,930 & 25,888,110 & 40,848,654 & 12,587,923 & 51,003,423 & 15,213,277 \\
        \hline
    \end{tabular}
    \caption{Average running times of Algorithms 1-7 for 10 instances for each problem size $n$ in microseconds; truncated normally distributed points.}
    \label{tab:ND-RT}\end{table}

In Figure \ref{fig:UD-RT-C} and Table \ref{tab:UD-RT} we consider the case of uniformly distributed points, while  Figure \ref{fig:ND-RT-C} Table \ref{tab:ND-RT} show the results of the experiment with truncated normally distributed points.
All axes use logarithmic scales. 

In the case of a uniform point distribution (Figure \ref{fig:UD-RT-C}), cppAPs turns out to be the real-time-wise fastest algorithm up to $n=2^{21}$, closely followed by cppMM. Only then cppAPs's asymptotical time complexity of $O(n^2)$ becomes (slowly) visible and it is overtaken by cppMM. For $n > 2^{21}$ up to the end of the range of our experiments, $n = 2^{25}$, cppMM performs best. cppDC performs very similar to its linear contenders cppRL, cppKM, and cppMM, despite its $O(n\log n)$ time complexity.

In the case of truncated normally distributed points (Figure \ref{fig:ND-RT-C}), the running time measurements yield qualitatively similar results,
with cppAPs fastest up to $n = 2^{23}$.

\begin{figure}[h]
    \centering
    \includegraphics[width=\curvewidth]{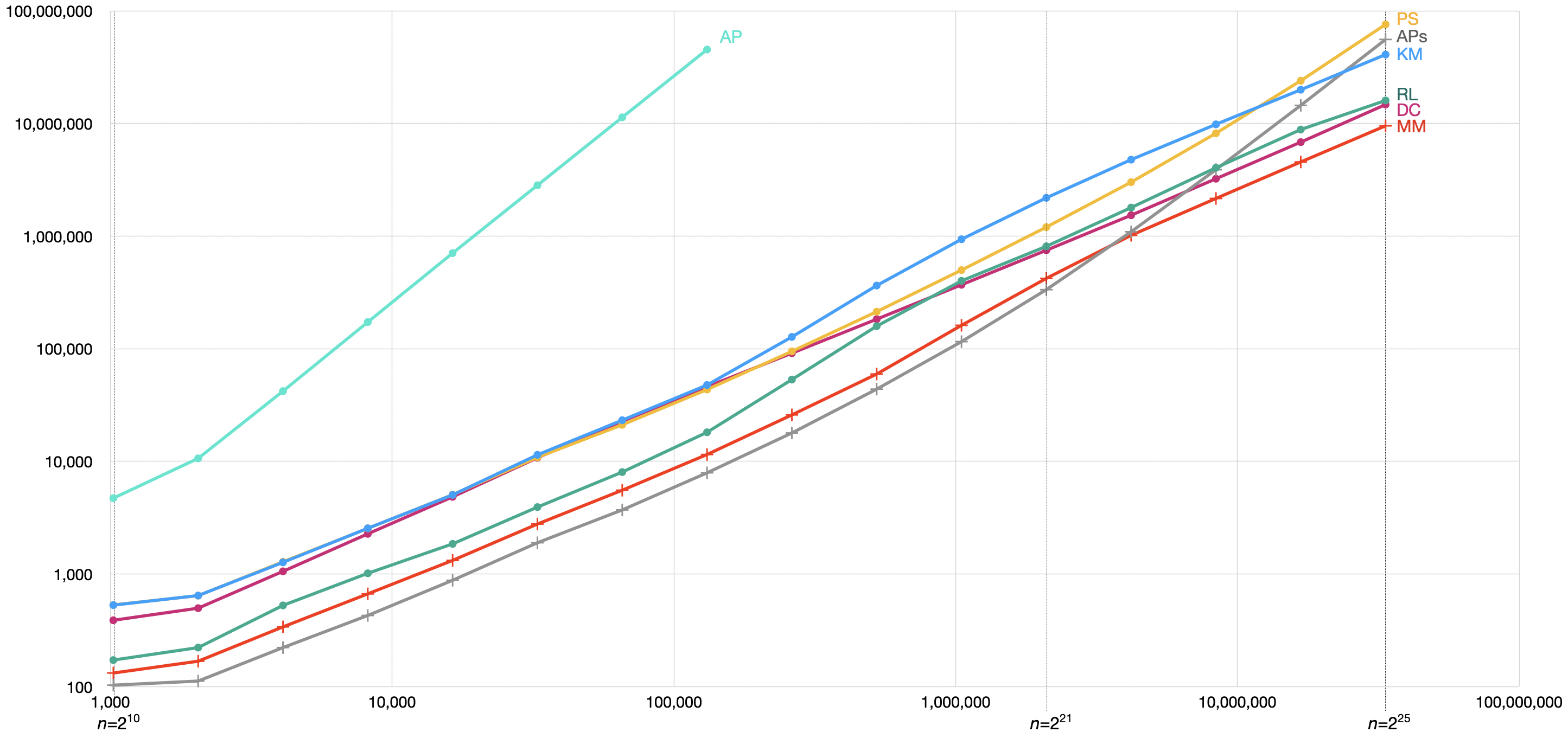}
    \caption{Running times of Algorithms 1-7 with uniformly distributed points (microseconds, log scales).}
    \Description{Seven curves show the running time of the algorithms for n = 1024 to $33\,554\,432$ points rising from the lower left to the upper right.}
    \label{fig:UD-RT-C}
\end{figure}

\begin{figure}[h]
    \centering
    \includegraphics[width=\curvewidth]{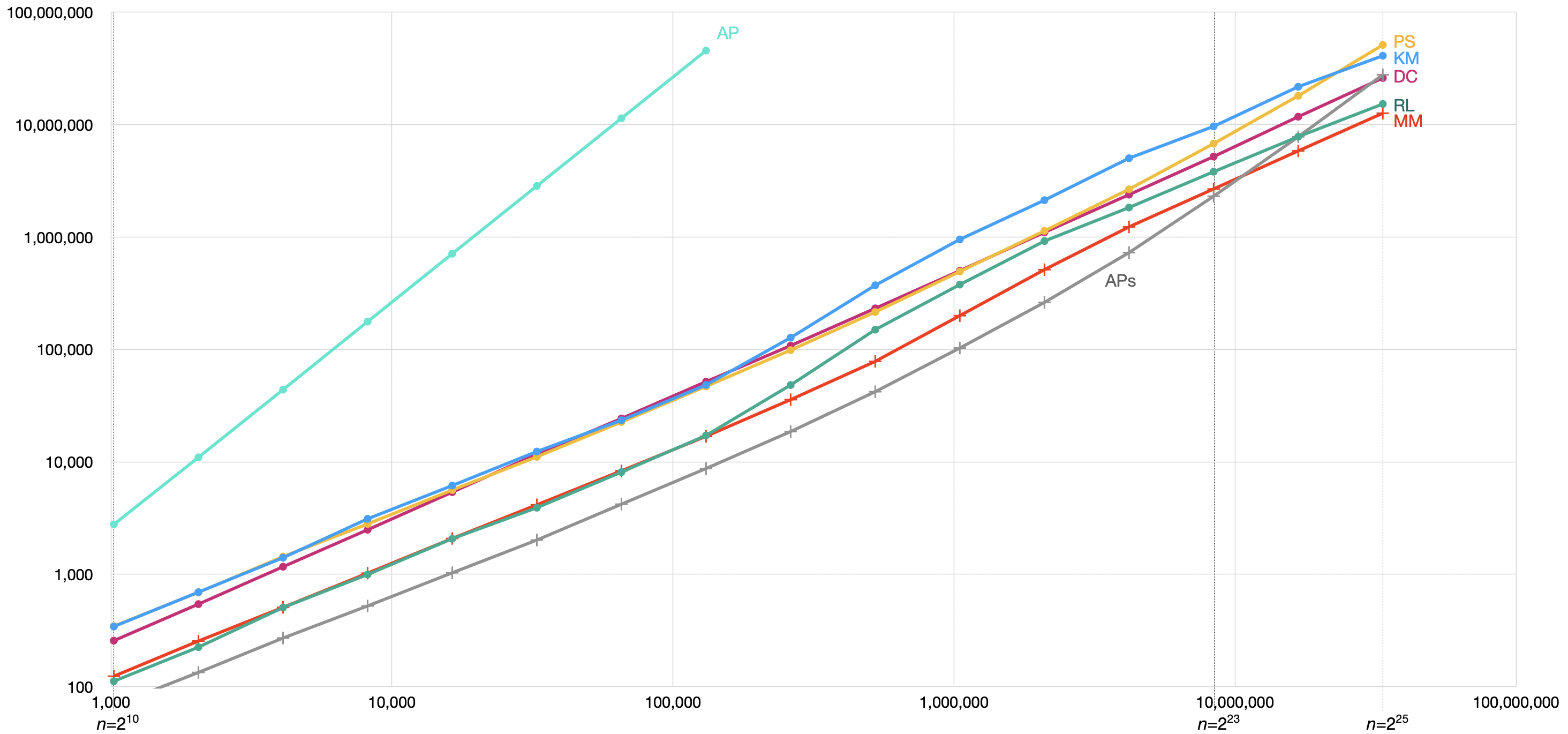}
    \caption{Running times of Algorithms 1-7 with truncated normally distributed points ($\mu = (0.5, 0.5)^\top$, $\sigma = 0.2$, microseconds, log scales).}
    \Description{Seven curves show the running time of the algorithms for n = 1024 to $33\,554\,432$ points rising from the lower left to the upper right, similar to Figure \ref{fig:UD-RT-C}.}
    \label{fig:ND-RT-C}
\end{figure}

The experiments show that cppMM outperforms all other algorithms in both cases of distribution, however overtaking cppAPs only at $n = 2^{22}$ and $2^{24}$, respectively. Compared to its fastest overall contender, cppRL, cppMM is easier to implement since there is no need to resort to a hash table simulating the two-dimensional grid, because the grid used in cppMM is of size $O(n)$ only and can therefore be implemented as a simple two-dimensional array, guaranteeing constant and linear access time to individual and all cells, respectively. In addition, it is a purely deterministic algorithm, in contrast to cppRL that uses a random approximation.

Of course, it is easy to construct a worst-case problem instance that leads to a running time of $ \Theta(n^2)$ of cppMM:~Put most points into only one of the squares of the $\bar{\delta}$-grid\footnote{Provided some points still reside near the boundary of the unit square as in Figure \ref{fig:Grid_n2}, since otherwise the grid is refined accordingly, see lines 2--4 in Algorithm \ref{alg:cppMM}.\label{FN:boundarypoints}}, so that cppMM is forced to compare all $O(n^2)$ pairs of points inside this square and therefore effectively turns into cppAP. Figure \ref{fig:Grid_n2} illustrates this scenario.

\begin{figure}[h]
    \centering
    \includegraphics[width=\gridwidth]{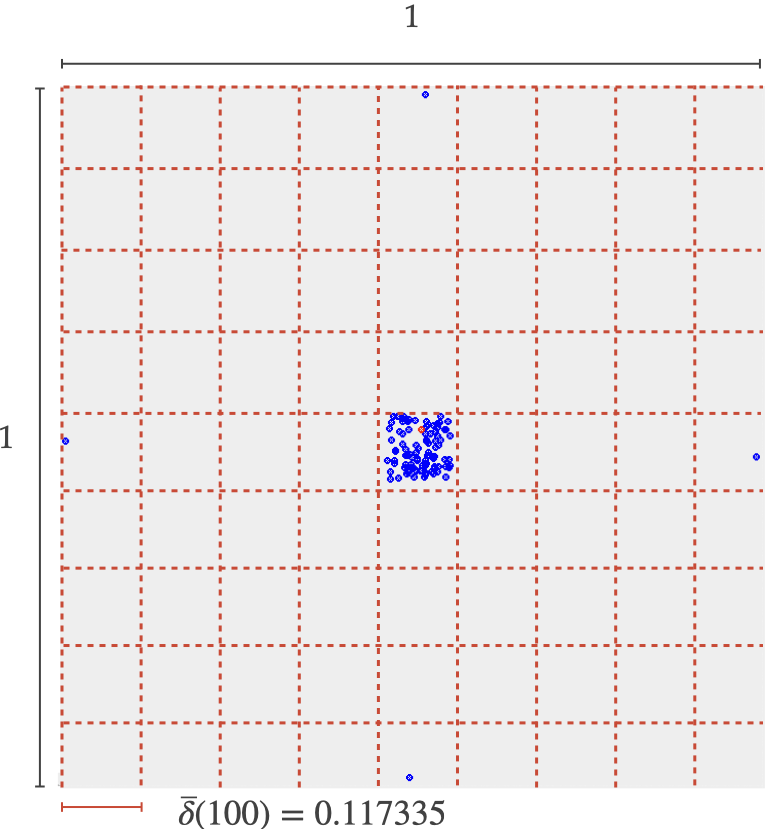}
    \caption{Distribution of $n = 100$ points in the unit square yielding $\Theta(n^2)$ running time of cppMM.} 
    \Description{All points reside in a central grid cell except four points, except for four points, each of which lies on one of the four boundaries of the unit square.}
\label{fig:Grid_n2}
\end{figure}

Next we study the average-case time complexity of cppMM. For uniformly distributed points we prove linear time complexity in Section \ref{sec:complexity}; truncated normally distributed points are discussed in the following.

As illustrated in Figure \ref{fig:sigma6}, there are `harmless' truncated normal distributions with, say, $\sigma \geq 0.25$, which sufficiently resemble the uniform distribution and thus do not significantly degrade the time complexity of cppMM compared to the uniform case. However, the worst-case scenario depicted in Figure \ref{fig:Grid_n2} will most likely not be attained with a normal distribution either:~Consider the case $\sigma=2^{-5}$ depicted in the lower right part of Figure \ref{fig:sigma6}. While this point cloud shares the property that most points are concentrated in the center with Figure \ref{fig:Grid_n2}, the probability of outliers on the boundary is very close to zero due to the small variance. In particular, while for $\sigma = 0.25$ and $n = 2^{24}$ the probability that a point lies in $[0,n^{-1/2}]\cup[n^{-1/2},1]\times[0,n^{-1/2}]\cup[n^{-1/2},1]$\footnote{Recall that $\bar \delta \in O(n^{-1/2})$.} is approximately 1, it is of magnitude $10^{-9}$ for $\sigma=2^{-5}$. Hence, we can readjust the grid accordingly\footnoteref{FN:boundarypoints}, and therefore do not expect a much worse behaviour than in the uniform case. This is confirmed by Figure \ref{fig:sigma6curves}.

\begin{figure}[h]
    \centering
    \includegraphics[width=0.9\textwidth]{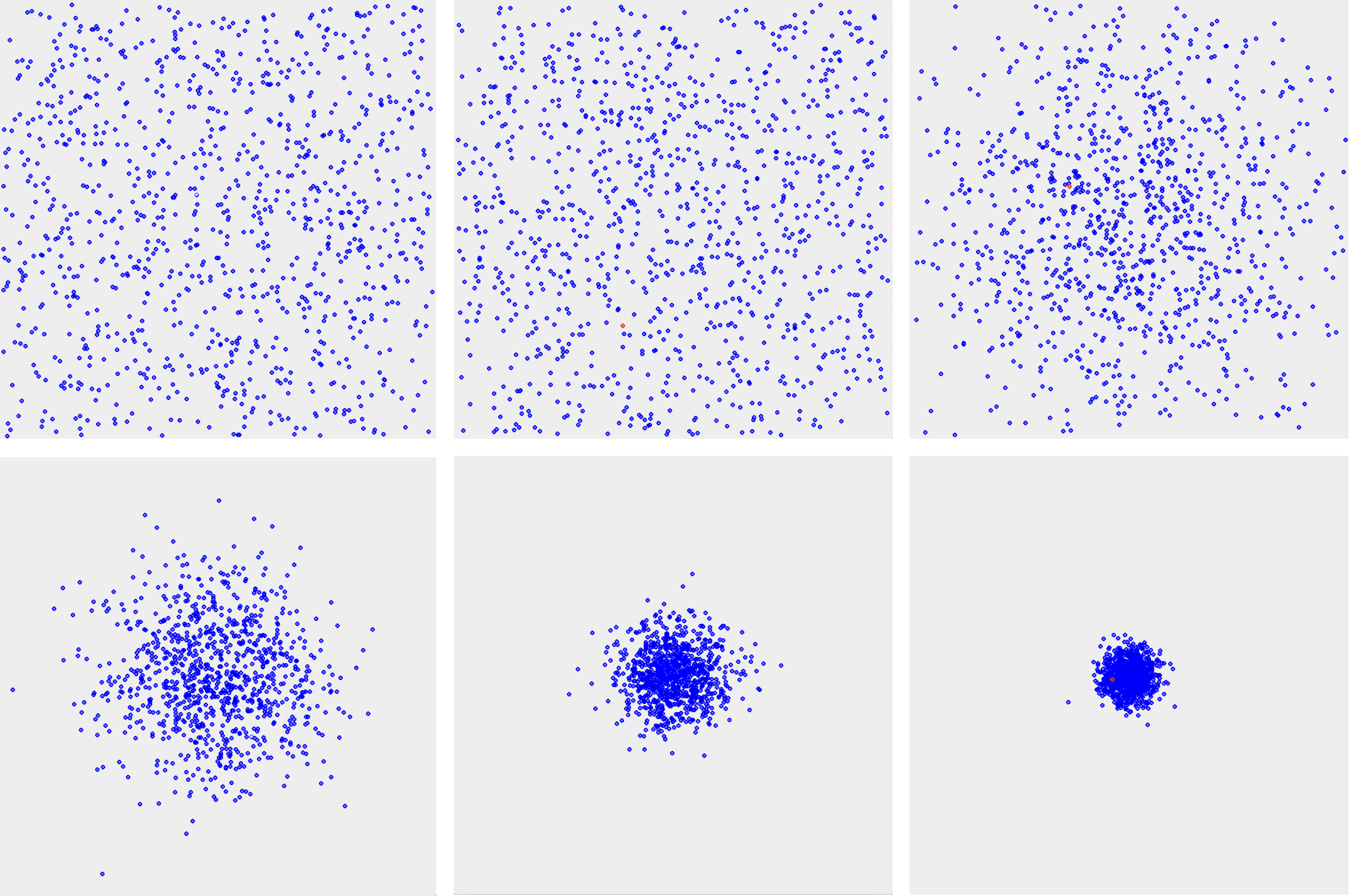}
    \caption{Truncated normally distributed point sets (n = 1024) with $\mu = (0.5, 0.5)^\top$ and $\sigma=2^{-k}$ with $k\in\{0,\ldots,5\}$ from the upper left to the lower right of the figure.}
    \Description{Six point clouds, starting with more or less equally scattered points, more and more concentrating the points in the central area.}
\label{fig:sigma6}
\end{figure}

In Figure \ref{fig:sigma6curves} we study the robustness of the running time of the algorithms considered for decreasing $\sigma$. For each value of $\sigma\in\{2^{-k}\mid k\in\{0,\dots,6\}\}$, ten truncated normally distributed random point sets with $2^{24}$ points were solved. For the total running time we observe that while cppMM is fastest for values $\sigma \geq 1/8$, it slows down a little as $\sigma$ decreases, but still outperforms cppDC, cppPS, and cppKM, and stays on par with cppRL. 

\begin{figure}[h]
    \centering
    \includegraphics[width=\textwidth]{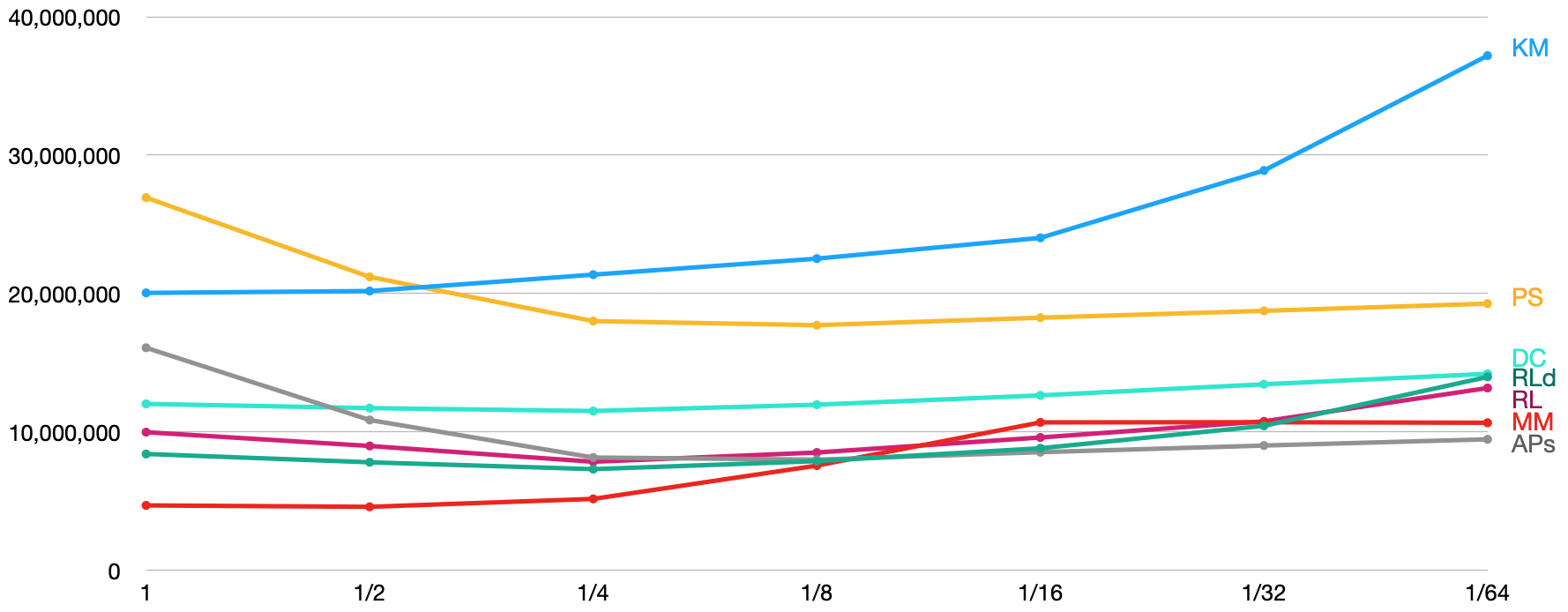}
    \caption{Running time of Algorithms 2-7 with truncated normally distributed point sets with $\sigma$ varying between 1 and 1/64 
    and n = $2^{24} = 16\,777\,216$ points (elapsed time for ten repetitions in microseconds, log scales).}
    \Description{Seven curves show the running time of the algorithms for n = 16777216 points when sigma is reduced from 1 to 1/64 by halving.}
    \label{fig:sigma6curves}
\end{figure}

\section{Complexity result}
\label{sec:complexity}

Finally, we are going to prove linear time complexity for cppMM in the case of a uniformly distributed point set $P$.

Intuitively, building the grid takes one iteration over $P$, thus $O(n)$ steps. The search for the minimum distance itself works as illustrated in Figure \ref{fig:Grid}, that is, each point is matched with partner-points in at most five adjacent grid cells. We have $O(n)$ grid cells for $n$ points. Due to the uniform distribution assumption, each cell is expected to host $O(1)$ points. Thus, the average overall running time is of the order $O(n) + O(n) \times 5 \times O(1) = O(n)$. We will formally prove this result.

\begin{theorem}\label{thm:uniform}
Let the probability space $(\Omega,\mathcal{F},\mathbb{P})$ carry all random variables we consider.
Let $n\in \mathbb{N} \setminus\{1\}$ be the number of points $P_1,\dots,P_n$ in the unit square, let $N=N_1+N_2$ be the total number of computational steps of Algorithm~\ref{alg:cppMM}, where $N_1$ is the number of steps of the pre-computation  (lines 1--10) and $N_2$ is the number of comparisons (lines 15--16).
Assume that the random variables $P_1,\dots,P_n$ are independent and uniformly distributed on the unit square.
Then the computational complexity of Algorithm~\ref{alg:cppMM} is of linear order, that is
$$\mathbb{E}[N]\in O(n).$$
\end{theorem}

\begin{proof}
Clearly, 
$$\mathbb{E}[N_1]\in O(n),$$
since we iterate over all $n = |P|$ points. Line 9 is $O(1)$, since adding an element to the bucket $G_{i,j}$ requires a constant number of steps.

Denote $M=\lceil 1/\bar\delta \rceil$
and let for all $i,j\in\{1,\dots,M\}$, $N_{i,j}$ be the number of points in the grid cell $G_{i,j}$.
In each grid cell $G_{i,j}$ that we visit, we compare all points $N_{i,j}$ in this grid cell with the other points in the grid cell, taking $N_{i,j}\cdot(N_{i,j}-1)/2$ operations, and with all points in four adjecent cells, taking $N_{i,j}\cdot(N_{i,j+1}+N_{i+1,j-1}+N_{i+1,j}+N_{i+1,j+1})$ operations. On the boundaries of the grid, we make less comparisons. In the following expression, the first double sum accounts for interior grid cells, while the single sums treat the boundaries of the unit square. We have
\begin{equation}\label{eq:N2}
\begin{aligned}
\mathbb{E}[N_2] 
& =
\mathbb{E}\Big[
\sum_{i=1}^{M-1}\sum_{j=2}^{M-1} N_{i,j}\cdot\Big(\frac{1}{2}(N_{i,j}-1)+N_{i,j+1}+N_{i+1,j-1}+N_{i+1,j}+N_{i+1,j+1}\Big)
\\&\quad+ \sum_{i=1}^{M-1} N_{i,1}\cdot\Big(\frac{1}{2}(N_{i,1}-1)+N_{i,2}+N_{i+1,1}+N_{i+1,2}\Big)
\\&\quad+\sum_{i=1}^{M-1} N_{i,M}\cdot\Big(\frac{1}{2}(N_{i,M}-1)+N_{i+1,M-1}+N_{i+1,M}\Big)
\\&\quad+\sum_{j=1}^{M-1} N_{M,j}\cdot\Big(\frac{1}{2}(N_{M,j}-1)+N_{M,j+1}\Big)
+ \frac{1}{2}N_{M,M}(N_{M,M}-1)
\Big]
\\&=
\sum_{i=1}^{M-1}\sum_{j=2}^{M-1}\Big(
\frac{1}{2}(\mathbb{E}[N_{i,j}^2]-\mathbb{E}[N_{i,j}])+\mathbb{E}[N_{i,j}\cdot N_{i,j+1}]
\\&\quad\quad\quad\quad\quad\quad+\mathbb{E}[N_{i,j}\cdot N_{i+1,j-1}]
+\mathbb{E}[N_{i,j}\cdot N_{i+1,j}]+\mathbb{E}[N_{i,j}\cdot N_{i+1,j+1}]\Big)
\\&\quad+ \sum_{i=1}^{M-1}\Big(
\frac{1}{2}(\mathbb{E}[N_{i,1}^2]-\mathbb{E}[N_{i,1}])
+
\mathbb{E}[N_{i,1}\cdot N_{i,2}]+\mathbb{E}[N_{i,1}\cdot N_{i+1,1}]
+\mathbb{E}[N_{i,1}\cdot N_{i+1,2}]\Big)
\\&\quad+\sum_{i=1}^{M-1} \Big(\frac{1}{2}(\mathbb{E}[N_{i,M}^2]-\mathbb{E}[N_{i,M}])
+\mathbb{E}[N_{i,M}\cdot N_{i+1,M-1}]
+\mathbb{E}[N_{i,M}\cdot N_{i+1,M}]\Big)
\\&\quad+\sum_{j=1}^{M-1}\Big(\frac{1}{2} (\mathbb{E}[N_{M,j}^2]-\mathbb{E}[N_{M,j}])
+\mathbb{E}[N_{M,j}\cdot N_{M,j+1}]\Big)
+ \frac{1}{2}(\mathbb{E}[N_{M,M}^2]-\mathbb{E}[N_{M,M}]).
\end{aligned}
\end{equation}
For all $i,j\in\{1,\dots,M\}$ we have that $N_{i,j}=\sum_{k=1}^n 1_{\{P_k\in G_{i,j}\}}$.
Since the $n$ points are uniformly distributed on the unit square and since for all $i,j\in\{1,\dots,M\}$, the area of the cell $G_{i,j}$ is $\bar\delta^2 \in O(1/n)$, we have
\begin{equation}\label{eq:ud_1}
 \mathbb{E}[N_{i,j}] = \sum_{k=1}^n \mathbb{P}(\{P_k\in G_{i,j}\})
 \in n \times O(1/n) = O(1).
\end{equation}

Next, for all $i,j\in\{1,\dots,M\}$,
\begin{equation}\label{eq:ud_2-step1}
\begin{aligned}
\mathbb{E}[N_{i,j}^2]
&=
\sum_{k=1}^n \sum_{\ell=1}^n \mathbb{E}[1_{\{P_k\in G_{i,j}\}} \cdot 1_{\{P_\ell\in G_{i,j}\}}]
=
\sum_{k=1}^n \Big( \mathbb{P}(\{P_k\in G_{i,j}\}) + \sum_{\substack{\ell=1\\ \ell \neq k}}^n \mathbb{E}[1_{\{P_k\in G_{i,j}\}\cap\{P_\ell\in G_{i,j}\}}]\Big)
\\&=
\mathbb{E}[N_{i,j}]+
\sum_{k=1}^n \sum_{\substack{\ell=1\\ \ell \neq k}}^n \mathbb{P}(\{P_k\in G_{i,j}\}\cap\{P_\ell\in G_{i,j}\}).
\end{aligned}
\end{equation}
Since for $\ell\neq k$, $P_{\ell}$ and $P_k$ are independent, we have
\begin{equation*}
\begin{aligned}
\mathbb{P}(\{P_k\in G_{i,j}\}\cap\{P_\ell\in G_{i,j}\})
&=
\mathbb{P}(\{P_k\in G_{i,j}\})
\cdot \mathbb{P}(\{P_\ell\in G_{i,j}\})
\in O(1/n) \times O(1/n).
\end{aligned}
\end{equation*}
Hence,
\begin{equation}\label{eq:ud_2}
\mathbb{E}[N_{i,j}^2]\in O(1) + n \times (n-1) \times O(1/n) \times O(1/n)
= O(1).
\end{equation}

The last class of summands that appear in \eqref{eq:N2} are $\mathbb{E}[N_{i,j}\cdot N_{\alpha,\beta}]$ with $(i,j)\neq(\alpha,\beta)$. For all $i,j,\alpha,\beta\in\{1,\dots,M\}$ with $(i,j)\neq(\alpha,\beta)$ we have
\begin{equation}\label{eq:ud_3-step1}
\begin{aligned}
\mathbb{E}[N_{i,j}\cdot N_{\alpha,\beta}]
&=
\sum_{k=1}^n \sum_{\ell=1}^n \mathbb{E}[1_{\{P_k\in G_{i,j}\}} \cdot 1_{\{P_\ell\in G_{\alpha,\beta}\}}]
\\&=
\sum_{k=1}^n \Big( \mathbb{E}[1_{\{P_k\in G_{i,j}\}\cap\{P_k\in G_{\alpha,\beta}\}}] + \sum_{\substack{\ell=1\\ \ell \neq k}}^n \mathbb{E}[1_{\{P_k\in G_{i,j}\}\cap\{P_\ell\in G_{\alpha,\beta}\}}]\Big)
\\&=
0+
\sum_{k=1}^n \sum_{\substack{\ell=1\\ \ell \neq k}}^n \mathbb{P}(\{P_k\in G_{i,j}\}\cap\{P_\ell\in G_{\alpha,\beta}\}).
\end{aligned}
\end{equation}
This ensures that
\begin{equation}\label{eq:ud_3}
\mathbb{E}[N_{i,j}\cdot N_{\alpha,\beta}]\in n \times (n-1) \times O(1/n) \times O(1/n)
= O(1).
\end{equation}
Recalling that $M=1/\lceil 1/\bar\delta \rceil\in O(\sqrt{n})$ and combining \eqref{eq:N2} with \eqref{eq:ud_1}, \eqref{eq:ud_2}, and \eqref{eq:ud_3} gives
\begin{equation*}
\mathbb{E}[N_2] \in O(\sqrt{n}) \times O(\sqrt{n}) \times O(1)
= O(n).
\end{equation*}
This closes the proof.
\end{proof}

\section{Conclusion}

While classical solutions to the CPP problem cannot be improved from a theoretical perspective, we found ample opportunities for improving the effective algorithmic runtime.
We presented two novel algorithms that are  deterministic and non-recursive, based on theoretical findings from mathematical optimal packing theory.
For cppMM and uniformly distributed points we proved $O(n)$ time complexity. 
We conducted a comparison study to classical approaches. The findings are summarized in
Table \ref{tab:overview}.

\begin{table}[ht]
\scriptsize
\caption{Summary of Algorithms 1-7.} 
\label{tab:overview}
\begin{tabular}{lcccccc}
    \toprule
    Algorithm & No. & Ease of                  & Average                       & Worst         & Sensitive                       & Measured \\
              &     & implementation           & case time                     & case time     & to $\sigma$                     & running time \\
              &     &                          & complexity                    & complexity    &            & (UD, ND with $\sigma \geq 1/4$) \\
              &     & I=easy -- III=hard &                               &               &  & I=low -- III=high \\
    \midrule
    cppAP     & 1   & I                        & $O(n^2)$                      &  $O(n^2)$     & no                               & III \\
    cppDC     & 2   & III                        & $O(n \log n)$                 & $O(n \log n)$ & no                               & II \\
    cppPS     & 3   & II                        & $O(n \log n)$                        & $O(n \log n)$        & no                               & II \\
    cppRL     & 4   & II                        & $O(n)$                        & $O(n)$        & no                               & II \\
    cppKM     & 5   & III                        & $O(n)$                        & $O(n)$        & empirically:~yes                               & II \\
    cppAPs    & 6   & I                        & $O(n^2)$ & $O(n^2)$      & no                               & I--II \\
    cppMM     & 7   & II                        & $O(n)$\footnotemark                        & $O(n^2)$      & not to a relevant extent                               & I \\
  \bottomrule
\end{tabular}
\end{table}
\footnotetext{Proven for uniformly distributed points, observable also for truncated normally distributed points.}

We observe that with respect to the average case asymptotic time complexity, cppMM (for uniformly distributed point sets), cppRL, and cppKM are on par, followed by cppDC and cppPS, and finally, by cppAPs and cppAP. However, as shown in Figure \ref{fig:for2}, up to very high values of $n$, the complexity of cppAPs is much closer to $O(n)$ than to $O(n^2)$ and only approaches the asymptotic value slowly and for very large values of $n$. 

Concerning the measured running time, cppMM dominates the other algorithms 
except for cppAPs, which is the absolute winner up to rather large problem sizes.

Most algorithms are insensitive to the choice of the point set.
This is empirically also confirmed for cppMM, while there exist constructed worst case examples where it needs $O(n^2)$ time.

With respect to the ease of implementation, we assigned the algorithms to classes I to III depending on the number of lines of code used and the auxiliary data structures or algorithms required. Class I contains the two extremely simple `nested loop' algorithms cppAP and cppAPs; cppPS, cppRL, and cppMM populate class II, while cppDC and cppKM exhibit the most complex implementations. 

Taken together, our algorithms are at least as fast as their contenders while being easy (cppMM) or even extremely easy (cppAPs) to implement. Of course, despite our best efforts, we cannot rule out the possibility that our measurements are biased by individual weaknesses in the implementations (see the source code at https://www.aau.at/en/isys/ias/research/). We would be grateful for any information in this regard.

\section*{Declarations}
\subsection*{Funding / Competing Interests}
This work was not funded externally. The authors have no competing interests to declare that are relevant to the content of this article.


\bibliography{cpp}

@book{CormenEtAl2003,
    title     = {Introduction to Algorithms},
    author    = {Thomas H. Cormen and Charles E. Leiserson and Ronald L. Rivest and Clifford Stein},
    edition   = {$2^{nd}$},
    year      = {2003},
    publisher = {The MIT Press},
    address   = {Cambridge, MA; London, England} 
}

@incollection{SzaboEtAl2001,
author="Szab{\'o}, P. G.
and Csendes, T.
and Casado, L. G.
and Garc{\'i}a, I.",
editor="Giannessi, Franco
and Pardalos, Panos
and Rapcs{\'a}k, Tam{\'a}s",
title="Packing Equal Circles in a Square I. --- Problem Setting and Bounds for Optimal Solutions",
bookTitle="Optimization Theory: Recent Developments from M{\'a}trah{\'a}za",
year="2001",
publisher="Springer US",
address="Boston, MA",
pages="191--206",
isbn="978-1-4613-0295-7",
doi="10.1007/978-1-4613-0295-7_14",
url="https://doi.org/10.1007/978-1-4613-0295-7_14"
}

@misc{Lipton2009,
author = {Richard J. Lipton}, 
title = {Rabin Flips a Coin},
year = {2009},
note    = {https://rjlipton.wpcomstaging.com/2009/03/01/rabin-flips-a-coin/ (last visited: December 10, 2021)}
}

@inproceedings{Rabin1976,
    author       = {Michael O. Rabin},
    title        = {Probabilistic Algorithms},
    booktitle    = {Algorithms and Complexity: New Directions and Recent Results. Proceedings of a Symposium on New Directions and Recent Results in Algorithms and Complexity Held by the Computer Science Department, Carnegie-Mellon University, April 7-9, 1976},
    year         = {1976},
    editor       = {J.~F.~Traub},
    volume       = {},
    number       = {},
    series       = {},
    pages = {21--38},
    month        = {},
    organization = {},
    isbn         = {9780126975406},
    publisher    = {Academic Press},
    address      = {New York, NY; San Francisco, CA; London, England} 
}

@article{KhullerMatias1995,
    author   = {Khuller, S. and Matias, Y.},
    title    = {A Simple Randomized Sieve Algorithm for the {C}losest-{P}air {P}roblem},
    journal  = {Information and Computation},
    year     = {1995},
    volume   = {118},
    number   = {},
    pages    = {34--37}
}

@article{HinrichsEtAl1988,
author = {Hinrichs, Klaus and Nievergelt, Jug and Schorn, Peter},
title = {Plane-sweep solves the closest pair problem elegantly},
year = {1988},
issue_date = {Jan. 11, 1988},
publisher = {Elsevier North-Holland, Inc.},
address = {USA},
volume = {26},
number = {5},
issn = {0020-0190},
url = {https://doi.org/10.1016/0020-0190(88)90150-0},
doi = {10.1016/0020-0190(88)90150-0},
journal = {Inf. Process. Lett.},
month = {jan},
pages = {255–261},
numpages = {7}
}

@book{PreparataShamos1985,
author = {Preparata, Franco P. and Shamos, Michael I.},
title = {Computational geometry: an introduction},
year = {1985},
isbn = {0387961313},
publisher = {Springer-Verlag},
address = {Berlin, Heidelberg}
}

@inproceedings{ShamosHoey1976,
author = {Shamos, Michael Ian and Hoey, Dan},
title = {Geometric intersection problems},
year = {1976},
publisher = {IEEE Computer Society},
address = {USA},
url = {https://doi.org/10.1109/SFCS.1976.16},
doi = {10.1109/SFCS.1976.16},
booktitle = {Proceedings of the 17th Annual Symposium on Foundations of Computer Science},
pages = {208–215},
numpages = {8},
series = {SFCS '76}
}

@String{Computer = "{IEEE} Computer" }

@String{Academic = "Academic Press" }

@String{Springer = "Springer-Verlag" }

\end{document}